\documentclass[fontsize=11pt, a4paper, oneside, DIV=12, toc=listofnumbered, toc=bibnumbered, numbers = noendperiod]{scrartcl}
\pdfoutput=1

\usepackage{cmap}
\usepackage[T1]{fontenc}

\usepackage{newunicodechar}
\newunicodechar{ρ}{\rho}
\newunicodechar{σ}{\sigma}
\newunicodechar{λ}{\lambda}
\newunicodechar{Λ}{\Lambda}
\newunicodechar{μ}{\mu}
\newunicodechar{ν}{\nu}
\newunicodechar{ψ}{\psi}
\newunicodechar{ϕ}{\phi}
\newunicodechar{φ}{\varphi}
\newunicodechar{π}{\pi}
\newunicodechar{α}{\alpha}
\newunicodechar{ϵ}{\epsilon}
\newunicodechar{ε}{\varepsilon}
\newunicodechar{δ}{\delta}

\newcommand{\cP}{\mathcal{P}}

\newcommand{\cN}{\mathcal{N}}

\newcommand{\bD}{{\textbf{D}}}

\newcommand{\nn}{\nonumber}
\newcommand{\R}{\mathbb{R}}

\newcommand{\kb}[1]{|#1\rangle\!\langle#1|} 

\usepackage{microtype}

\usepackage[ngerman, english]{babel}
\usepackage{csquotes}

\usepackage{xparse}
\usepackage{xspace}
\usepackage{makecell}
\usepackage{titling}
\usepackage{appendix}

\usepackage{versions}
\excludeversion{DRAFT}
\excludeversion{TRASH}

\usepackage{amssymb}
\usepackage{amsmath}
\usepackage{mathrsfs}

\usepackage{authblk}

\usepackage{cool}
\usepackage{xcolor}
\usepackage{hyperref}
\hypersetup{
	colorlinks,
	linkcolor={red!50!black},
	citecolor={red!50!black},
	urlcolor={blue!80!black}
}
\usepackage[nameinlink]{cleveref}

\usepackage{amsthm}
\newtheorem{theorem}{Theorem}
\newtheorem{lemma}[theorem]{Lemma}
\newtheorem{proposition}[theorem]{Proposition}

\theoremstyle{definition}
\newtheorem{remark}[theorem]{Remark}

\newtheorem*{definition*}{Definition}

\usepackage{graphicx}         
\setkeys{Gin}{width=\textwidth} %
\usepackage[multidot,spae]{grffile} 
\usepackage[compatibility=false]{caption}
\usepackage{subcaption}
\graphicspath{{./}{img/}}

\makeatletter
\let\x@caption\caption
\newcommand{\x@@caption}[2][\empty]{%
	\ifx\empty#1\relax\x@caption{#2}%
	\else\x@caption[#1]{\textbf{#1.} #2}%
	\fi}
\let\caption\x@@caption
\renewcommand{\captionabove}[2][\empty]{\captionsetup{position=above}%
	\x@@caption[#1]{#2}%
}
\renewcommand{\captionbelow}[2][\empty]{\captionsetup{position=below}%
	\x@@caption[#1]{#2}%
}
\makeatother

\usepackage[backend=biber, sorting=none]{biblatex}
\DeclareFieldFormat{citehyperref}{%
	\DeclareFieldAlias{bibhyperref}{noformat}
	\bibhyperref{#1}}

\DeclareFieldFormat{textcitehyperref}{%
	\DeclareFieldAlias{bibhyperref}{noformat}
	\bibhyperref{%
		#1%
		\ifbool{cbx:parens}
		{\bibcloseparen\global\boolfalse{cbx:parens}}
		{}}}

\savebibmacro{cite}
\savebibmacro{textcite}

\renewbibmacro*{cite}{%
	\printtext[citehyperref]{%
		\restorebibmacro{cite}%
		\usebibmacro{cite}}}

\renewbibmacro*{textcite}{%
	\ifboolexpr{
		( not test {\iffieldundef{prenote}} and
		test {\ifnumequal{\value{citecount}}{1}} )
		or
		( not test {\iffieldundef{postnote}} and
		test {\ifnumequal{\value{citecount}}{\value{citetotal}}} )
	}
	{\DeclareFieldAlias{textcitehyperref}{noformat}}
	{}%
	\printtext[textcitehyperref]{%
		\restorebibmacro{textcite}%
		\usebibmacro{textcite}}}

\AtEveryBibitem{
	\clearfield{note}
	\clearfield{issn}
	\clearfield{url}
	\clearfield{urlyear}
	\clearfield{urlmonth}
	\clearfield{eprintclass}
}

\DeclareBibliographyCategory{cited}
\AtEveryCitekey{\addtocategory{cited}{\thefield{entrykey}}}

\usepackage{suffix}
\usepackage{xparse}
\usepackage{braket}

\usepackage{physics}

\renewcommand{\tr}{\Tr}

\usepackage{mathtools}
\usepackage{amssymb}
\usepackage{tensor}
\usepackage{centernot}
\usepackage[binary-units=true]{siunitx}

\newcommand{\br}[1]{\left(#1\right)}

\NewDocumentCommand{\stared}{m m}{\IfBooleanTF{#2}{#1*}{#1}}
\NewDocumentCommand{\es}{m o}{
	\IfNoValueTF{#2}{
		\mathbb{\uppercase{#1}}^{\lowercase{#1\/}}
	}{
		\mathbb{\uppercase{#1}}^{#2\/}
	}
}

\newcommand{\reals}{\mathbb{R}}

\newcommand{\HS}{\mathcal{H}}

\NewDocumentCommand{\BO}{o}{\mathcal{B}\br{\IfValueTF{#1}{#1}{\HS}}}

\NewDocumentCommand{\DM}{o}{\mathcal{D}\br{\IfValueTF{#1}{#1}{\HS}}}

\NewDocumentCommand{\pos}{o}{\mathscr{P}(\IfValueTF{#1}{#1}{\HS})}
\NewDocumentCommand{\posdef}{o}{\mathscr{P}_{\mkern -4mu + \mkern -1mu}(\IfValueTF{#1}{#1}{\HS})}

\NewDocumentCommand{\supp}{m}{\textrm{supp}\br{#1}}

\DeclarePairedDelimiterX{\infdivx}[2]{(}{)}{
	#1\delimsize|\delimsize|#2
}

\DeclareDocumentCommand{\qrd}{o o m m}{\ensuremath{\IfValueTF{#1}{#1}{D}\IfValueTF{#2}{_{#2}}{}\/\infdivx*{#3}{#4}}}
\newcommand{\grd}{\qrd[\widehat{D}][α]}

\newcommand{\renyi}{Rényi\xspace}

\NewDocumentCommand{\ent}{d()g}{\ensuremath{H\br{\IfValueTF{#1}{#1}{#2}}}}
\NewDocumentCommand{\qent}{d()g}{\ensuremath{H\br{\IfValueTF{#1}{#1}{#2}}}}

\newcommand*\ric[1]{\vphantom{#1}\smash{#1_{}\kern-\scriptspace}}

\DeclarePairedDelimiter{\mnorm}{\lVert}{\rVert}

\title{The $α \to 1$ Limit of the Sharp Quantum \renyi Divergence}
\author{Bjarne Bergh}
\author{Robert Salzmann}
\author{Nilanjana Datta}
\affil{DAMTP, University of Cambridge, United Kingdom}

\begin{document}
	\maketitle
	\begin{abstract}
		Fawzi and Fawzi \cite{fawzi_defining_2020} recently defined the sharp \renyi divergence, $D_\alpha^\#$, for $\alpha \in (1, \infty)$, as an additional quantum \renyi divergence with nice mathematical properties and applications in quantum channel discrimination and quantum communication. One of their open questions was the limit $α \to 1$ of this divergence. By finding a new expression of the sharp divergence in terms of a minimization of the geometric \renyi divergence, we show that this limit is equal to the Belavkin-Staszewski relative entropy. Analogous minimizations of arbitrary generalized divergences lead to a new family of generalized divergences that we call {\em{kringel divergences}}, and for which we prove various properties including the data-processing inequality. 
	\end{abstract}
	\section{Geometric and sharp R\'enyi divergences} 
	Let $\HS$ be a complex finite-dimensional Hilbert space, and $\BO$ the set of linear operators on $\HS$. We write $\pos$ for the set of positive semi-definite operators on $\HS$ and $\posdef$ for the set of positive definite operators. Let $\DM$ denote the set of density matrices, i.e.~the set of positive semi-definite operators with trace 1. For $A, B \in \pos$ we further write $A \ll B$ if $\supp{A} \subseteq \supp{B}$.

	For $A, B \in \posdef$ and $α \in \reals$ the \emph{weighted matrix geometric mean} is defined as
	
	\begin{equation}\label{def:matrix-geom-mean}
	    A \#_α B \coloneqq A^{\frac{1}{2}}\left(A^{-\frac{1}{2}}B A^{-\frac{1}{2}}\right)^α A^{\frac{1}{2}} \,.
	\end{equation}
	For $α \ge 0$ this definition can be easily extended to more general $A, B \in \pos$, $B \ll A$, by restricting the Hilbert space to the support of $A$ (this also corresponds to directly using pseudo-inverses in (\ref{def:matrix-geom-mean})). For $α \in [0,1]$ the weighted matrix geometric mean satisfies many desirable properties of a matrix mean \cite{kubo_means_1979}. 
	
	The geometric \renyi divergence (also called the maximal \renyi divergence), first introduced by Matsumoto~\cite{matsumoto_new_2013}, can be defined in terms of the weighted matrix geometric mean. For $ρ, σ \in \pos$, $ρ \ll σ$ and $α \in (0, \infty)$ define the \emph{geometric trace function} as 
	\begin{equation}
	    \label{def:geometric-trace-fn}
		\qrd[\widehat{Q}][α]{ρ}{σ} \coloneqq \tr (\sigma \#_{\alpha} \rho) = \tr(σ^{\frac{1}{2}}(σ^{-\frac{1}{2}}ρ σ^{-\frac{1}{2}})^α σ^{\frac{1}{2}})\,.
	\end{equation}
	This can be extended to general $ρ, σ \in \pos$ by setting
	\begin{equation}
	   \qrd[\widehat{Q}][α]{ρ}{σ} \coloneqq \lim_{ε \to 0} \qrd[\widehat{Q}][α]{ρ}{σ + ε \IdentityMatrix} \, .
	\end{equation}
	For $α > 1$ this limit is equal to $+ \infty$ if $ρ \centernot\ll σ$, whereas for $α \in (0, 1]$ the limit is always finite and explicit expressions for it can be found in~\cite{matsumoto_new_2013, hiai_different_2017, katariya_geometric_2020}.
	For $α \in (0, 1) \cup (1, \infty)$ one then defines the \emph{geometric \renyi divergence} as \cite{matsumoto_new_2013, petz_contraction_1998,  fang_geometric_2019, katariya_geometric_2020}
	\begin{equation}\label{def:geometric-div}
		\qrd[\widehat{D}][α]{ρ}{σ} \coloneqq \frac{1}{α - 1} \log \qrd[\widehat{Q}][α]{ρ}{σ}\, .
	\end{equation}
	
	It reduces to the classical $α$-\renyi divergence for commuting states and satisfies the data-processing inequality for $α \in (0, 1) \cup (1, 2]$~\cite{matsumoto_new_2013}. Further, it is known~\cite{matsumoto_new_2013, tomamichel_quantum_2016, katariya_geometric_2020} that for $ρ \in \DM$ and $σ \in \pos$
		\begin{equation}
		  \lim\limits_{α \to 1} 	\qrd[\widehat{D}][α]{ρ}{σ}  = \widehat{D}(\rho||\sigma) \coloneqq  \tr(ρ \log(ρ^{\frac{1}{2}}σ^{-1}ρ^{\frac{1}{2}})),
		\end{equation}
	the Belavkin-Staszewski relative entropy \cite{belavkin_c-algebraic_1982}. 
	
	In~\cite{fawzi_defining_2020}, Fawzi and Fazwi defined the sharp trace function for $α \in (1, \infty)$:
	\begin{equation}\label{def:sharp-trace-fn}
		\qrd[Q^\#][α]{ρ}{σ} = \min_{A \geq 0}\Set{ \tr A | ρ \leq σ\#_{\frac{1}{α}} A }.
	\end{equation}
    They defined the sharp R\'enyi divergence of order $\alpha$ in terms of it as follows:
	\begin{equation}\label{def:sharp-div}
		\qrd[D^\#][α]{ρ}{σ} \coloneqq \frac{1}{α - 1} \log(\qrd[Q^\#][α]{ρ}{σ}),
	\end{equation}
	and proved that it satisfies the data-processing inequality, and that it also reduces to the classical $α$-\renyi divergence for commuting states. Further, they showed that this divergence has several desirable computational
properties such as an efficient semidefinite programming representation for states and quantum channels (i.e.~linear, completely positive trace-preserving maps), and a crucial chain rule property
which can be exploited to obtain important information-theoretic results concerning quantum channel discrimination and quantum channel capacities.\smallskip

A natural question to ask is: {\em{What is $\lim_{\alpha \to 1} \qrd[D^\#][α]{ρ}{σ}$?}} 
\smallskip

This was left as an open question in~\cite{fawzi_defining_2020}.
In this paper we answer this question by proving that this limit is given by the Belavkin-Staszewski relative entropy (see Theorem~\ref{α1:theorem-sharp-bs-limit} below).

\smallskip
To address the limit of $α \to 1$ we want to make use of the following alternative characterization of the sharp trace function and sharp R\'enyi divergence.
    \begin{proposition} \label{prop:key}  
		For $α \in (1,\infty)$ and $ρ, σ \in \pos$
		\begin{align}\label{sharp-min-geometric}
		 \qrd[Q^\#][α]{ρ}{σ} &= \min_{A \geq ρ} \qrd[\widehat{Q}][α]{A}{σ}, \\
		 \qrd[D^\#][α]{ρ}{σ} &= \min_{A \geq ρ} \qrd[\widehat{D}][α]{A}{σ}.
		 \end{align}
	\end{proposition}
	\begin{remark}
	Note that the $A$ in the above expressions are in general unnormalized states and we use the definitions of the geometric trace function and the geometric R\'enyi divergence as in (\ref{def:geometric-trace-fn}) and (\ref{def:geometric-div}) without additional normalization factors.
	\end{remark}
	For what is going to follow, we prove a slightly stronger version of the above statement, given by the following lemma.
	
	\begin{lemma}\label{α1:compact-restriction}
		For $α \in (1, a]$, with $a \in (1, \infty)$, and $ρ, σ \in \pos$,
		\begin{equation}
			\qrd[Q^\#][α]{ρ}{σ} = \min_{A \geq ρ} \qrd[\widehat{Q}][α]{A}{σ},
		\end{equation}
		and the minimization can be restricted to $A \in K$ where $K$ is a compact subset of $\Set{A \in \pos| A \ll σ}$ depending only on $ρ, σ, a$, but not on $α$.
	\end{lemma}
	\begin{proof}
		Recall the definition of $\qrd[Q^\#][α]{ρ}{σ}$ for $\alpha >1$~\cite{fawzi_defining_2020}:
		\begin{equation}
			\qrd[Q^\#][α]{ρ}{σ} = \min_{A \geq 0}\Set{\tr A | ρ \leq σ^{\frac{1}{2}} (σ^{-\frac{1}{2}} A σ^{-\frac{1}{2}})^{\frac{1}{α}} σ^{\frac{1}{2}}}\,.
		\end{equation}
	    Fawzi and Fawzi~\cite{fawzi_defining_2020} showed that the minimization can be further restricted to $0 \leq A \leq c^{α - 1} \tr(ρ) \Pi_σ$, where $c = \mnorm{σ^{-\frac{1}{2}} ρ σ^{-\frac{1}{2}}}$ is the spectral norm, and $\Pi_σ$ is the projection onto the support of $σ$. It is easy to see that if $ρ \centernot\ll σ$ the minimization is infeasible, and so $\qrd[Q^\#][α]{ρ}{σ} = +\infty$. Moreover, if $ρ \centernot\ll σ$ and $A \ge ρ$ then also $A \centernot\ll σ$, hence $\min_{A \ge ρ}\qrd[\widehat{Q}][α]{A}{σ} = +\infty$ and the statement holds. Thus, we can assume $ρ \ll σ$.
		Taking $C \coloneqq \tr (ρ)\sup{_{α \in [1, a]}} c^{α - 1}$, which is always finite, we get:
		\begin{align}
			\qrd[Q^\#][α]{ρ}{σ} &= \min_{A \geq 0}\Set{\tr A | ρ \leq σ^{\frac{1}{2}} (σ^{-\frac{1}{2}} A σ^{-\frac{1}{2}})^{\frac{1}{α}} σ^{\frac{1}{2}}, \; A \leq C \Pi_σ} \\
			&=\min_{A \geq 0}\Set{\tr(σ^{\frac{1}{2}} A σ^{\frac{1}{2}}) | ρ \leq σ^{\frac{1}{2}} A^{\frac{1}{α}} σ^{\frac{1}{2}}, \; A \leq C σ^{-\frac{1}{2}}\Pi_σ σ^{-\frac{1}{2}} \leq C \mnorm{σ^{-1}} \Pi_σ }\\
			&= \min_{A \geq 0} \Set{\tr(σ^{\frac{1}{2}} A^α σ^{\frac{1}{2}}) | ρ \leq σ^{\frac{1}{2}} A σ^{\frac{1}{2}}, \; A^α \leq C \mnorm{σ^{-1}} \Pi_σ }
		\end{align}
		where we redefined $A$ multiple times, and understand $\mnorm{σ^{-1}}$ as the operator norm of the pseudo-inverse.
		While for general matrices $A, B \geq 0$, the statements $A^α \leq B$ and $A \leq B^\frac{1}{α}$ are not equivalent, they are equivalent in the case in which $B$ is a constant times a projector onto a subspace that includes the support of $A$. 
		This is because in this case $B$ is diagonal in the same basis in which $A$ (and hence also $A^\alpha$) is diagonal. 
		Hence, the operator inequality $A^α \leq B$ turns into a condition only on the eigenvalues of $A$, which is then equivalent to $A \leq B^\frac{1}{α}$. Thus, we finally get:
		\begin{align}
			\qrd[Q^\#][α]{ρ}{σ} &= \min_{A \geq 0} \Set{\tr(σ^{\frac{1}{2}} A^α σ^{\frac{1}{2}}) | ρ \leq σ^{\frac{1}{2}} A σ^{\frac{1}{2}}, \; A \leq (C \mnorm{σ^{-1}})^\frac{1}{α} \Pi_σ } \\
			&= \min_{A \geq 0}   \Set{\qrd[\widehat{Q}][α]{A}{σ} | ρ \leq A , \; A \leq  (C \mnorm{σ^{-1}})^\frac{1}{α} σ^{\frac{1}{2}} \Pi_σ σ^{\frac{1}{2}} \leq (C \mnorm{σ^{-1}})^\frac{1}{α}\mnorm{σ}\Pi_σ } 
		\end{align}
		With $\tilde{C} := \sup_{α \in [1, a]} (C \mnorm{σ^{-1}})^\frac{1}{α}\mnorm{σ} $, which exists, and $K := \set{A\in \pos | ρ \leq A \leq \tilde{C} \Pi_σ}$, which is compact, we have for all $α \in (1,b]$:
		\begin{equation}
			\qrd[Q^\#][α]{ρ}{σ} = \min_{A \in K} \qrd[\widehat{Q}][α]{A}{σ}\, .
		\end{equation}
	\end{proof}	
	
	Since the logarithm on $\reals$ is monotone, Lemma~\ref{α1:compact-restriction} implies that also $\qrd[D^\#][α]{ρ}{σ} = \min_{A \geq ρ} \qrd[\widehat{D}][α]{A}{σ}$ for all $α > 1$. This completes the proof of Proposition~\ref{prop:key}. Note, that for $α > 2$ the geometric divergence does not satisfy the data-processing inequality, while the sharp divergence does.

	\section{The $α \to 1$ limit}
	We show in Theorem~\ref{α1:theorem-sharp-bs-limit} at the end of this section that the limit $α \to 1$ of the sharp divergence is the Belavkin-Staszewski relative entropy. 
	The key step in the proof is to show that $\dv{α} \qrd[Q^\#][α]{ρ}{σ} |_{α = 1}$ exists and is equal to $\dv{α} \qrd[\widehat{Q}][α]{ρ}{σ} |_{α = 1}$.
	One way to establish this is to use the following theorem~\autocite[Theorem~2.2.1]{demyanov_theory_1971}:
	\begin{theorem}[Dem'yanov and Malozemov (1971)]\label{α1:demyanov-theorem}
		Let $U \subset \reals^n$ be open, $K \subset \reals^m$ compact, and $f \colon U\times K \to \reals$ continuous and also that $\nabla_u f(u,k)$ is (jointly) continuous. Then, the function $g(u) = \min_{k \in K} f(u,k)$ has for every $u \in U$ a one-sided directional derivative along every $v \in \reals^n$, which can be computed as
		\begin{equation}
			\lim\limits_{t \searrow 0} \frac{g(u + v t) - g(u)}{t} = \min_{k \in R(u)} \langle \nabla_u f(u, k), v \rangle,
		\end{equation}
	where $R(u) = \set{k \in K | f(u,k) = g(u)}$.\footnote{In \cite{demyanov_theory_1971} the theorem is phrased with a maximum instead of a minimum, but it is easy to see that this is equivalent upon setting $f \mapsto -f$.}
	\end{theorem}
The following lemmas establish the properties of $\widehat{Q}_α$ which are necessary in order to apply Theorem~\ref{α1:demyanov-theorem}.

    \begin{lemma}\label{lemma:matrix-continuity}
        Let $\Omega$ be a subset of $\reals^n$ and $h\colon \Omega \times [0, \infty) \to [0, \infty)$ be a jointly continuous function. Then, the corresponding matrix function $h \colon \Omega \times \pos \to \pos$ is jointly continuous on $\Omega \times \pos$. 
    \end{lemma}
    \begin{proof}
        For a fixed $α \in \Omega$, the continuity of the matrix function $h(α, ρ)$  in $ρ$ follows from the continuity of $h(α, x)$ in $x$, for $x \in [0, \infty)$, by \cite[Theorem 6.2.37]{horn_topics_1991}. To see joint continuity, let $ρ_n \to ρ$ be a converging sequence in $\pos$ and $α_n \to α$ be a converging sequence in $\Omega$, as $n \to \infty$. Then, taking operator norms,
		\begin{align}\label{α1:continuity_triangle}
			\mnorm{h(α_n, ρ_n) - h(α, ρ)} \leq \mnorm{h(α_n, ρ_n) - h(α, ρ_n)} + \mnorm{h(α, ρ_n) - h(α, ρ)} \,.
		\end{align}
	    The second term on the right hand side goes to zero as $n \to \infty$ by the continuity of $h(α, ρ)$ in $ρ$. For the first term, note that $h(α_n, ρ_n)$ and $h(α, ρ_n)$ commute and the operator norm can be evaluated as
		\begin{equation}\label{α1:continuity_eigenvalues}
			\max_{\lambda \in \text{spec}(ρ_n)} |h(α_n, λ) - h(α, λ)| \,.
		\end{equation}
		Since the limit $α_n \to α$ exists, $α_n$ is eventually bounded, so there exists a compact subset $K$ of $\Omega$, such that $α_n \in K$ for all large enough $n$. 
		Analogously, the limit $ρ_n \to ρ$ exists, and so the spectrum of $ρ_n$ is eventually bounded, and hence there exists a compact subset $K'$ of $[0, \infty)$ such that $\text{spec}(ρ_n) \subset K'$ for all large enough $n$. Since $K\times K'$ is compact, $h$ is in fact uniformly continuous on $K \times K'$. Hence, the expression in (\ref{α1:continuity_eigenvalues}) goes to zero as $n \to \infty$ uniformly in all elements of the maximum and also uniformly in $ρ_n$. This shows that 
		\begin{equation}
		    \lim_{n \to \infty} \mnorm{h(α_n, ρ_n) - h(α, ρ)} = 0
		\end{equation}
		which completes the proof. 
	\end{proof}

	\begin{lemma}\label{α1:posdef-continuity}
	For $σ \in \pos$ fixed, the functions $f_ρ(α) = f(α, ρ) = \qrd[\widehat{Q}][α]{ρ}{σ}$ and $f'_ρ(α) = f'(α, ρ) = \dv{α} \qrd[\widehat{Q}][α]{ρ}{σ}$ are jointly continuous on $(0, \infty) \times \Set{A \in \pos| A \ll σ}$.
	\end{lemma}
	\begin{proof}
	    We have:
		\begin{align}
			\qrd[\widehat{Q}][α]{ρ}{σ} &= \tr(σ^{\frac{1}{2}}\qty(σ^{-\frac{1}{2}}ρ σ^{-\frac{1}{2}})^α σ^{\frac{1}{2}}) \\
			\dv{α}\qrd[\widehat{Q}][α]{ρ}{σ} &= \tr(σ^{\frac{1}{2}}\log(σ^{-\frac{1}{2}}ρ σ^{-\frac{1}{2}})\qty(σ^{-\frac{1}{2}}ρ σ^{-\frac{1}{2}})^α σ^{\frac{1}{2}})\, .
		\end{align}
		It is easy to check that the real functions $(α, x) \mapsto x^α$ and $(α, x) \mapsto x^α \log(x)$ are jointly continuous on $(0, \infty) \times [0, \infty)$. Also, for $A \in \Set{A \in \pos| A \ll σ}$ the expression $σ^{-\frac{1}{2}} A σ^{-\frac{1}{2}}$ is continuous in $A$ (remember that the inverses are all pseudo-inverses).  Hence, the continuity follows from Lemma \ref{lemma:matrix-continuity}, the continuity of the matrix product and the continuity of the trace.
	\end{proof}
	\begin{theorem}\label{α1:theorem-sharp-bs-limit}
		For $ρ\in \DM$ and $σ \in \pos$:
		\begin{equation}\label{α1:equation_theorem_limit}
			\lim\limits_{α \to 1} \qrd[D^\#][α]{ρ}{σ} = \widehat{D}(\rho||\sigma) 
		\end{equation}
	the Belavkin-Staszewski relative entropy. 
	\end{theorem}
	\begin{proof}
		As the sharp R\'enyi divergence, $D_\alpha^{\#}$, is only defined for $α > 1$, all there is to prove is a limit from above. Let us fix $ρ \in \DM$ and $σ \in \pos$.
		We can restrict our proof to the case $ρ \ll σ$, since otherwise $\qrd[D^\#][α]{ρ}{σ} = \widehat{D}_\alpha(\rho||\sigma) = \infty$ for all $α > 1$ and the statement clearly holds. 
		By Lemma \ref{α1:compact-restriction} there exists a compact set $K \subseteq \Set{A\in \pos| A \ll σ}$\footnote{$K$ depends on $ρ$ and $σ$, but since they are fixed in the entire proof we do not make this dependence explicit.} such that for $α \in (1, 2]$,
		\begin{equation}
		    \qrd[Q^\#][α]{ρ}{σ} = \min_{A \in K} \qrd[\widehat{Q}][α]{A}{σ} \, .
		\end{equation}
		For $α \in (0, \infty)$,\footnote{To apply Theorem \ref{α1:demyanov-theorem} we require an open interval in $α$ which includes 1.} we define
		\begin{align}
			f(α, A) &\coloneqq \qrd[\widehat{Q}][α]{A}{σ} \\
			g(α) &\coloneqq \min_{A \in K} f(α, A) = \min_{A \in K} \qrd[\widehat{Q}][α]{A}{σ} \,.
		\end{align}
		By Lemma \ref{α1:posdef-continuity}, $f(α, A)$ is jointly continuous on $(0, \infty) \times K$ and has a jointly continuous derivative in $α$. Hence by Theorem \ref{α1:demyanov-theorem}, $g(α)$ has the the following one-sided derivatives:
		\begin{align}
			\lim_{α \searrow 1} \frac{g(α) - g(1)}{α - 1} &= \min_{A \in R(α)} \dv{α} f(α, A) \label{α1:theorem-limit:limit-above} \\
			\lim_{α \nearrow 1} \frac{g(α) - g(1)}{α - 1} &= -\min_{A \in R(α)} \left(-\dv{α} f(α, A)\right)\label{α1:theorem-limit:limit-below}\, ,
		\end{align}
		where $R(α) := \set{A \in K | f(α, A) = \min_{B \in K} f(α, B)}$. Recall that the sharp \renyi divergence is only defined for $α > 1$ and so all we really need here is the limit from above (\ref{α1:theorem-limit:limit-above}). However, establishing that this is also equal to the limit from below makes things slightly simpler, as we are then able to directly use the chain rule later in (\ref{eq:limit-theorem:chain-rule}).
		
		For $α = 1$ the set $R(α)$ only contains $ρ$, since
		\begin{align}
		    \min_{A \geq ρ} \qrd[\widehat{Q}][1]{A}{σ} = \min_{A \geq ρ} \Tr(A) = \Tr(ρ),
		\end{align}
		and $A \geq \rho $, together with $\Tr(A) = \Tr(ρ)$, imply that $A = ρ$. 
		Hence, the two one-sided derivatives are equal and
		\begin{align}
		\label{eq:deriequal}
            \dv{α} g(α)\bigg|_{α = 1} = \dv{α} f(α, ρ) \bigg|_{α = 1} = \dv{α} \qrd[\widehat{Q}][α]{ρ}{σ}\bigg|_{α = 1}\,.
		\end{align}
		We further have $g(α) = \qrd[Q^\#][α]{ρ}{σ}$ for $α \in (1, 2]$, and $g(1) = \tr(ρ) = 1$, so
		\begin{align}
			\lim\limits_{α \searrow 1} \qrd[D^\#][α]{ρ}{σ} &= \lim\limits_{α \searrow 1} \frac{\log \qrd[Q^\#][α]{ρ}{σ}}{α - 1} = \lim\limits_{α \searrow 1} \frac{\log g(α) - \log g(1)}{α - 1} \nonumber\\
			&= \dv{α} \log g(α) \bigg|_{α = 1}  = \frac{\dv{α} g(α) \big|_{α = 1}}{g(1)} = \dv{α} g(α) \bigg|_{α = 1} \,. \label{eq:limit-theorem:chain-rule}
		\end{align}
	As the same argument also gives \begin{align}\lim\limits_{α \to 1} \qrd[\widehat{D}][α]{ρ}{σ} = \dv{α} \widehat Q_\alpha(\rho\|\sigma) \bigg|_{α = 1}, \end{align} we get by using  \eqref{eq:deriequal}
	\begin{equation}\label{α1:limit-theorem:final_eq}
		\lim\limits_{α \to 1} \qrd[D^\#][α]{ρ}{σ} = \lim\limits_{α \to 1} \grd{ρ}{σ} = \widehat{D}(\rho||\sigma)\, .
	\end{equation}
	\end{proof}
\section{Kringel divergences and their properties}

 A key ingredient of our main result, Theorem~\ref{α1:theorem-sharp-bs-limit}, was Proposition~\ref{prop:key}, which allowed us to express
the sharp divergence as a minimization of the geometric R\'enyi divergence. Analogous minimizations of arbitrary generalized divergences lead to an interesting new family of generalized divergences which we call {\em{kringel divergences}}\footnote{Since the symbols $D$, $\widetilde{D}_\alpha$, $\widehat{D}_\alpha$, $D^\flat_α$ and $D^\#_\alpha$ all refer to existing divergences, and generalized divergences are usually denoted as $\bD$, we use the symbol $\bD^\circ$ to denote this new family; hence the name
{\em{kringel divergences}} (kringel = circle in German).
  }. We introduce them in this section and prove some of their properties, including the data-processing inequality.
\medskip

	For a function $\bD: \pos\times\pos \to \R\cup\{-\infty,\infty\}$  we define for $\rho,\sigma\in\pos$
	\begin{align}
    \bD^\circ(\rho\|\sigma) = \inf_{A\ge\rho}\bD(A\|\sigma).
    \end{align}
    Moreover, we say  $\bD: \pos\times\pos \to \R\cup\{-\infty,\infty\}$  is a {\em{generalized divergence}} if it satisfies the data-processing inequality, i.e.~for any $\rho,\sigma\in\pos$ and any quantum channel $\cN$, we have
	\begin{align}
	\label{eq:DPI}
	\bD(\cN(\rho)\|\cN(\sigma)) \le \bD(\rho\|\sigma).
	\end{align}
	\begin{remark}
	Note that the above definition of the generalized divergence is an extension of the standard definition (see e.g.~in~\cite{khatri_principles_2020}). In the latter, the generalized divergence is considered as a function $\bD: \mathcal{D}(\mathcal{H}) \times\pos \to \R\cup\{\infty\}$, where $\mathcal{D}(\mathcal{H})$ denotes the set of density matrices (quantum states) on $\mathcal{H}$.
	\end{remark}
	\smallskip
	
	\noindent
	The following lemma shows that if $\bD$ is a generalized divergence then so is $\bD^\circ$. In this case we call $\bD^\circ$ the \emph{kringel divergence} of $\bD$.
 
\begin{lemma}[Data-processing inequality] \label{lem:DPI}
	Let $\bD$ be a generalized divergence, $\rho,\sigma\in\pos$ and $\cN$ be a quantum channel. Then
	\begin{align}
	\bD^\circ(\cN(\rho)\|\cN(\sigma)) \le \bD^\circ(\rho\|\sigma).
	\end{align}
\end{lemma}
\begin{proof}
	We have
	\begin{align}
	\nn\bD^\circ(\cN(\rho)\|\cN(\sigma)) &= \inf_{A\ge\cN(\rho)}\bD(A\|\cN(\sigma)) \le \inf_{\cN(A)\ge\cN(\rho)} \bD(\cN(A)\|\cN(\sigma))\\&\le\inf_{\cN(A)\ge\cN(\rho)} \bD(A\|\sigma)\le \inf_{A\ge\rho}\bD(A\|\sigma) = \bD^\circ(A\|\sigma).
	\end{align}
	Here, for the first inequality we have used the fact that the minimum increases when we only optimise over operators of the form $\cN(A)$ and for the second inequality we have used the data-processing inequality for $\bD$. For the third inequality we have used the fact that $A\ge\rho$ implies $\cN(A)\ge\cN(\rho)$, which follows by positivity and linearity of $\cN$.
\end{proof}
\begin{remark}
By Proposition~\ref{prop:key} we know that for $\alpha\in(1,\infty)$ we have
\begin{align}
\widehat D^\circ_\alpha = D^\#_\alpha,
\end{align}
where the left hand side denotes the kringel divergence corresponding to the geometric \renyi divergence \eqref{def:geometric-div}. Note that even though the geometric \renyi divergence satisfies the data-processing inequality only for $\alpha\in(0,1)\cup(1,2]$, by \cite[Proposition 3.2]{fawzi_defining_2020-1} its corresponding kringel divergence actually satisfies the data-processing inequality for all $\alpha\in(0,1)\cup(1,\infty).$ 
\end{remark}
\begin{lemma}
	\label{lem:subadd}
	If $\bD$ is subadditive then so is $\bD^\circ$, i.e. for all $\rho_1,\rho_2,\sigma_1,\sigma_2\in\pos$ we have
	\begin{align}
	\bD^\circ(\rho_1\otimes\rho_2\|\sigma_1\otimes\sigma_2) \le \bD^\circ(\rho_1\|\sigma_1) + \bD^\circ(\rho_2\|\sigma_2).
	\end{align}
\end{lemma}
\begin{proof}
	We have
	\begin{align}
	\nn \bD^\circ(\rho_1\otimes\rho_2\| \sigma_1\otimes\sigma_2) &= \inf_{A\ge\rho_1\otimes\rho_2}\bD(A\| \sigma_1\otimes\sigma_2) \le \inf_{A_1\otimes A_2\ge\rho_1\otimes\rho_2}\bD(A_1\otimes A_2\| \sigma_1\otimes\sigma_2) \\ &\le\nn  \inf_{A_1\otimes A_2\ge\rho_1\otimes\rho_2} \Big(\bD(A_1\|\sigma_1) + \bD(A_2\|\sigma_2)\Big)  \\&= \inf_{A_1\ge \rho_1} \bD(A_1\|\sigma_1)  +\inf_{A_2\ge \rho_2} \bD(A_2\|\sigma_2)\nn\\& = \bD^\circ(\rho_1\|\sigma_1) + \bD^\circ(\rho_2\|\sigma_2).
	\end{align}
    The third line above follows from the fact that $A_1\ge \rho_1$ and $A_2\ge\rho_2$ implies that $A_1\otimes A_2\ge\rho_1\otimes\rho_2$ because $A_1\otimes A_2 - \rho_1\otimes\rho_2 = \left(A_1-\rho_1\right)\otimes A_2 + \rho_1\otimes\left(A_2 - \rho_2\right)$.
\end{proof}

\subsection{Kringel divergences for $\alpha$-\renyi divergences}
We say $\bD_\alpha: \pos\times\pos \to \R\cup\{-\infty,\infty\}$ is a \emph{quantum generalization of the $\alpha$-\renyi divergence} if it reduces to the corresponding classical \renyi divergence if both entries commute. That is, if $\rho,\sigma\in\pos$ are commuting operators then
\begin{align}
\bD_\alpha(\rho\|\sigma) = \frac{1}{\alpha-1}\log\Big(\sum_i p_i^{\alpha} q_i^{1-\alpha}\Big),
\end{align}
where the $\{p_i\}_i$ and $\{q_i\}_i$ are eigenvalues of $\rho$ and $\sigma$ respectively with respect to a simultaneous eigenbasis.
The following lemma states that, for $\alpha>1$, if $\bD_\alpha$ is a quantum generalization of the \renyi relative entropy, then $\bD^\circ_\alpha$ too reduces to the classical \renyi divergence in the commuting case, and hence is itself a quantum generalization of the \renyi relative entropy.
\begin{lemma}
\label{lem:Reny}
Let $\alpha>1$ and $\bD_\alpha$ be a quantum generalization of the $\alpha$-\renyi divergence satisfying the data-processing inequality \eqref{eq:DPI}. Then also $\bD^\circ_{\alpha}$ is a quantum generalization of the $\alpha$-\renyi divergence.
\end{lemma}
\begin{proof}
Let $\rho, \sigma \in \pos $ such that $[\rho, \sigma]=0$. Then there exists a simultaneous eigenbasis $\{\ket{i}\}_i$ of $\rho$ and $\sigma$ such that 
\begin{align}
\rho = \sum_i p_i\kb{i},\quad\quad \sigma = \sum_{i}q_i\kb{i}.
\end{align} Clearly,
\begin{align}
\bD_\alpha^\circ(\rho\|\sigma) \le \bD_\alpha(\rho\|\sigma) = \frac{1}{\alpha-1}\log\Big(\sum_i p_i^{\alpha} q_i^{1-\alpha}\Big).
\end{align}
For the reversed inequality let $\cP$ be the pinching map defined as $\cP(\cdot) := \sum_{i}\kb{i}\cdot\kb{i}$. Hence, denoting for $A\ge \rho$ the diagonal entries by $a_i=\bra{i}A\ket{i}\ge p_i$, we see
\begin{align}
\nn \bD_\alpha^\circ(\rho\|\sigma) &= \inf_{A\ge\rho}\bD_\alpha(A\|\sigma) \ge\inf_{A\ge\rho}\bD_\alpha(\cP(A)\|\cP(\sigma)) 
\\&=\inf_{A\ge\rho} \frac{1}{\alpha-1}\log\Big(\sum_{i} a_i^{\alpha}q_i^{1-\alpha} \Big)=  \frac{1}{\alpha-1}\log\Big(\sum_{i} p_i^{\alpha}q_i^{1-\alpha}\Big),
\end{align}
where we have used data-processing inequality for $\bD_\alpha$ in the first inequality.
\end{proof}

\begin{remark}
Note that for $\alpha\in(0,1)$ and $\bD_\alpha$ being a quantum generalization of the $\alpha$-\renyi divergence, one easily sees that
\begin{align}
\bD_{\alpha}^\circ(\rho\|\sigma) = -\infty,
\end{align} 
for all $\rho,\sigma\in\pos$ with $\sigma\neq0$,
which is why we excluded this range of $\alpha$ in Lemma~\ref{lem:Reny} above and Proposition~\ref{prop:KringelMingel} below.
\end{remark}
Important quantum generalizations of the $\alpha$-R\'enyi divergence include the geometric- and sharp R\'enyi divergences, as well as the {\em{Petz R\'enyi divergence}} ($ D_\alpha$) \cite{petz_quasi-entropies_1986} and the {\em{sandwiched \renyi divergence}} ($\widetilde D_\alpha$) \cite{muller-lennert_quantum_2013, wilde_strong_2014}. The latter two are defined, respectively, as follows: for $\rho, \sigma \in \pos$
\begin{align}
D_\alpha (\rho||\sigma) &:= \frac{1}{\alpha -1} \log \tr \left(\rho^\alpha \sigma^{1-\alpha}\right), \nonumber\\
 \widetilde D_\alpha (\rho||\sigma) &:= \frac{1}{\alpha -1} \log \tr \left(\sigma^{\frac{1-\alpha}{2 \alpha}}\rho \sigma^{\frac{1-\alpha}{2 \alpha}}\right)^\alpha
\end{align}
if $\alpha\in(0,1)$, or $\alpha\in(1,\infty)$ and $\rho \ll\sigma$. If $\alpha\in(1,\infty)$ and $\rho\centernot\ll\sigma$, then $D_\alpha(\rho\|\sigma)=\widetilde D_\alpha(\rho\|\sigma)=\infty$. Both the Petz- and the sandwiched \renyi divergences are additive. Moreover, $D_\alpha$ for $\alpha\in(0,1)\cup(1,2]$, and $\widetilde D_\alpha$ for $\alpha\in[1/2,1)\cup (1,\infty)$, satisfy the data-processing inequality \eqref{eq:DPI}, as was shown
in \cite{petz_quasi-entropies_1986,tomamichel_fullyquantum_2009} and \cite{frank_monotonicity_2013}, respectively. Therefore, for both of these divergences the following proposition applies in the corresponding ranges of $\alpha$.
\begin{proposition}
    \label{prop:KringelMingel}
	Let $\alpha>1$ and $\bD_\alpha$ be a quantum generalization of the $\alpha$-\renyi divergence which is subadditive and satisfies the data-processing inequality~\eqref{eq:DPI}. Then
	\begin{align}
	\label{eq:bounds}
	\widetilde D_\alpha \le \bD^\circ_\alpha \le D^\#_\alpha.
	\end{align}
	In particular, this gives for any $\rho,\sigma\in\pos$
	\begin{align}
	\label{eq:reg}
	\lim_{n\to\infty} \frac{1}{n}\bD^\circ_\alpha(\rho^{\otimes n}\|\sigma^{\otimes n}) = 	\widetilde D_\alpha(\rho\|\sigma).
	\end{align}
	Moreover, in the case of the sandwiched divergence 
	\begin{align}
	\label{eq:sandwichKringel}
	\widetilde D^\circ_\alpha  = 	\widetilde D_\alpha. 
	\end{align}
\end{proposition}
\begin{proof}
As $\bD_\alpha$ satisfies the data-processing inequality, we know that $\bD_\alpha \le \widehat D_\alpha$; this follows from the argument~\cite[Section 4.2.3]{tomamichel_quantum_2016} (also see \cite{matsumoto_new_2013} where the argument originally appeared)\footnote{Note that in
in~\cite{tomamichel_quantum_2016, matsumoto_new_2013} the argument is presented for the case in which $\rho$ and $\sigma$ are states. However, it can be easily seen that, by a slight modification, the argument also works for the case $\rho,\sigma\in\pos$.}. This gives
\begin{align}
\bD^\circ_\alpha(\rho\|\sigma) = \min_{A\ge\rho}\bD_\alpha(A\|\sigma) \le \min_{A\ge\rho}\widehat D_\alpha(A\|\rho)  = \widehat D^\circ_\alpha(\rho\|\sigma) =  D^\#_\alpha(\rho\|\sigma).
\end{align}
Moreover, as $\alpha>1$ and additionally $\bD_\alpha$ is subadditive, Lemmas~\ref{lem:DPI}, \ref{lem:subadd} and \ref{lem:Reny} give that also $\bD^\circ_\alpha$ is a subadditive quantum generalization of the $\alpha$-\renyi divergence satisfying the data-processing inequality. Therefore, by the argument in~\cite[Section 4.2.2]{tomamichel_quantum_2016} (note that actually only subadditivity instead of additivity is used there) we have
\begin{align}
\widetilde D_\alpha \le \bD^\circ_\alpha,
\end{align}
which gives \eqref{eq:bounds}. Since, trivially, $\widetilde D^\circ_\alpha  \le	\widetilde D_\alpha$, \eqref{eq:sandwichKringel} follows immediately.

Lastly, using $\alpha>1$ again and~\cite[Proposition 3.4]{fawzi_defining_2020} gives for any $\rho,\sigma\in\pos$
\begin{align}
\lim_{n\to\infty} \frac{1}{n}D^\#_\alpha(\rho^{\otimes n}\|\sigma^{\otimes n}) = 	\widetilde D_\alpha(\rho\|\sigma),
\end{align}
which together with the additivity of the sandwiched R\'enyi divergence and \eqref{eq:bounds} gives \eqref{eq:reg}.

\end{proof}

\bigskip
\noindent
\textbf{Acknowledgements.}
The authors would like to thank the Institute for Pure and Applied Mathematics (IPAM) at UCLA for hosting a stimulating workshop on {\em{Entropy Inequalities, Quantum Information and Quantum Physics}} (February 8-11, 2021). The main question addressed in this paper was also posed as an open question at the workshop. The authors would also like to thank Eric Carlen and Mark~M.~Wilde for helpful comments on earlier versions of this paper. Bjarne Bergh is supported by the UK Engineering and Physical Sciences Research Council (EPSRC), grant number EP/V52024X/1. Robert Salzmann is supported by the Cambridge Commonwealth, European and International Trust. 

\medskip
\noindent
\textbf{Data availability.}
Data sharing is not applicable to this article as no new data were created or analyzed in this study.
\printbibliography
\end{document}